\documentclass[11pt]{article}
\usepackage[a4paper,margin=1in]{geometry}
\usepackage[T1]{fontenc}
\usepackage[utf8]{inputenc}
\usepackage{lmodern}
\usepackage{amsmath,amssymb,amsthm,mathtools}
\usepackage[ruled,vlined,linesnumbered]{algorithm2e}
\usepackage{microtype}
\usepackage{needspace}
\usepackage{flafter}
\usepackage{placeins}
\usepackage[round,authoryear]{natbib}
\usepackage{xurl}
\usepackage[hidelinks]{hyperref}
\hypersetup{pdftitle={Budget-Independent Influence Maximization in Nearly Linear Time},pdfauthor={Zhijie Zhang},pdflang={en-US}}
\theoremstyle{plain}
\newtheorem{theorem}{Theorem}[section]
\newtheorem{lemma}[theorem]{Lemma}

\theoremstyle{definition}
\newtheorem{definition}[theorem]{Definition}
\theoremstyle{remark}

\SetKwInput{KwIn}{Input}
\SetKwInput{KwOut}{Output}
\SetKw{Return}{return}
\SetKwComment{Comment}{$\triangleright$ }{}
\DontPrintSemicolon
\SetAlgoVlined
\SetAlFnt{\small}
\SetAlCapFnt{\small}
\SetAlCapNameFnt{\small\bfseries}
\SetAlgoSkip{medskip}
\SetAlgoNlRelativeSize{-1}
\newcommand{\eps}{\varepsilon}
\newcommand{\OPT}{\operatorname{OPT}}
\newcommand{\RR}{\operatorname{RR}}
\newcommand{\Reach}{\operatorname{Reach}}
\newcommand{\E}{\mathbb E}
\newcommand{\Prob}{\mathbb P}
\newcommand{\ind}{\mathbf 1}
\newcommand{\HIT}{\mathsf{HIT}}
\newcommand{\Absorb}{\mathsf{AbsorbRR}}
\newcommand{\Select}{\mathsf{SeedSelection}}
\newcommand{\Greedy}{\mathsf{BucketGreedy}}
\newcommand{\Score}{\mathsf{ScoreRR}}
\newcommand{\Main}{\mathsf{BudgetIndependentIM}}
\newcommand{\cR}{\mathcal R}
\newcommand{\cA}{\mathcal A}

\title{Budget-Independent Influence Maximization \\ in Nearly Linear Time}
\author{Zhijie Zhang\\
Fuzhou University\\
\texttt{zzhang@fzu.edu.cn}}
\date{}
\begin{document}
\maketitle
\begin{abstract}
Influence maximization asks for $k$ seed vertices that maximize the expected spread of a diffusion process in a network. Standard near-optimal-time algorithms based on reverse-reachable sampling achieve a $(1-1/e-\eps)$ approximation, but their expected running-time bounds grow linearly with the seed budget $k$. We remove this multiplicative dependence: for the independent cascade model, our algorithm succeeds with probability at least $1-\delta$ in $O((m+n)\eps^{-3}\log(2n/\delta))$ expected time. The result extends to triggering models with explicitly charged local sampling costs.

We reserve $O(\eps k)$ seed positions for cost-weighted random vertices, allowing reverse-reachable searches to stop as soon as they encounter a reserved seed. An independent sample-count estimation phase uses a statistic that also controls the expected search cost. Matching these quantities eliminates the multiplicative dependence on $k$ while preserving the approximation guarantee.
\end{abstract}

\section{Introduction}
\label{sec:intro}
Influence maximization asks how a limited number of initial activations can reach a large population through network diffusion. It captures the possibility of reaching people beyond those directly recruited, a motivation for network-aware viral marketing~\citep{DR01,RD02} and peer-led health-information outreach~\citep{Wilder18}. Because different seeds may reach overlapping audiences, the value of one seed depends on the others selected. The task is therefore to allocate a limited seeding budget jointly, rather than simply rank vertices by individual prominence~\citep{KKT03,KKT15}.

Formally, the input consists of a directed graph $G=(V,E)$ with $n$ vertices and $m$ edges, a stochastic diffusion model, and a seed budget $k$. For a seed set $S$, let $\sigma(S)$ denote the expected number of activated vertices; the goal is to maximize $\sigma(S)$ subject to $|S|\le k$. Under the independent cascade (IC) and triggering models, the objective is monotone and submodular, although the optimization problem is NP-hard~\citep{KKT03,KKT15}. The classical greedy analysis yields a $1-1/e$ approximation when exact marginal values are available~\citep{NWF78}.

The challenge is to obtain a comparable guarantee efficiently. Repeated marginal-influence evaluations can require many reachability computations over random graph realizations. Reverse influence sampling addresses repeated evaluation by reducing influence maximization to maximum coverage on sampled reverse-reachable (RR) sets~\citep{BBCL14,TIM14}. Its standard running-time bounds, however, retain a multiplicative dependence on the seed budget. For instance, the expected time of TIM~\citep{TIM14} is
\begin{equation}
 O\!\left((m+n)\eps^{-2}
       \left[k\log n+\log\frac1\delta\right]\right)
 \label{eq:tim-bound}
\end{equation}
for a $(1-1/e-\eps)$ approximation with failure probability $\delta$. The revised full analysis of~\citet{BBCL16} also contains a linear dependence on $k$. These bounds are nearly linear for fixed $k$, but can become polynomially larger when the budget grows with the network. For example, with fixed accuracy, inverse-polynomial failure probability, and $k=\lfloor\sqrt n\rfloor$, the bound in~\eqref{eq:tim-bound} becomes $\widetilde O((m+n)\sqrt n)$ rather than $\widetilde O(m+n)$. Eliminating the budget factor would allow larger seed sets without this additional polynomial cost.

We establish a nearly linear expected-time guarantee for every input budget by reducing the cost of generating RR sets. We reserve $O(\eps k)$ seed positions for vertices drawn independently in proportion to their reverse-search costs and require the final solution to contain them. A reverse search can then stop as soon as it hits a reserved seed, because its sample is already covered by the eventual output. Whereas the Hit-and-Stop (HIST) algorithm constructs its stopping set through empirical influence optimization and validation~\citep{GWWC20,GWWLT22}, our method uses independent cost-weighted sampling to bound search cost directly. The approximation guarantee relies only on the reserved set being small, not on its influence.

To determine the number of RR samples, we use an independent \emph{sample-count estimation phase} based on a bounded measure of reverse-search cost. Its mean appears inversely in the sample-count bound and linearly in the expected cost bound for each stopped search, so it cancels when the two bounds are combined. Together with the saving from the reserved seeds, this yields the following guarantee.

\begin{theorem}[Influence maximization in nearly linear time]
\label{thm:main-ic}
In the adjacency-list model, for every IC network, $1\le k\le n$, $0<\eps\le1/4$, and $0<\delta<1/2$, there is a randomized algorithm returning $S\subseteq V$ with $|S|\le k$ such that
\begin{equation}
 \Prob\!\left[\sigma(S)\ge(1-1/e-\eps)\OPT_k\right]\ge1-\delta,
 \qquad \OPT_k=\max_{|T|\le k}\sigma(T).
 \label{eq:main-quality}
\end{equation}
Its expected running time is
\begin{equation}
 O\!\left(\frac{m+n}{\eps^3}
       \left[\ln\frac{en}{k}+\frac1k\ln\frac4\delta\right]\right)
 \subseteq O\!\left((m+n)\eps^{-3}\log\frac{2n}{\delta}\right).
 \label{eq:main-time}
\end{equation}
\end{theorem}

Compared with~\eqref{eq:tim-bound}, the theorem removes the multiplicative budget factor while changing the precision dependence from $\eps^{-2}$ to $\eps^{-3}$. We develop the analysis for IC and extend it to triggering models with explicit local sampling-cost assumptions in Section~\ref{sec:trigger-extension}. A time-capped version gives a deterministic running-time limit with one additional logarithmic repetition factor.

\paragraph{Parallel work and use of AI.}
In parallel work, \citet{Seddighin26} obtains a $(1-1/e-\eps)$-approximation for influence maximization in $\widetilde O_{\eps}(m+n)$ time, likewise removing the multiplicative dependence on the seed budget. He reports having circulated a draft of this result to researchers on June~3, 2026. We credit Seddighin with the first nearly linear-time, budget-independent algorithm achieving this approximation guarantee for arbitrary input budgets. The two projects were developed independently: neither author was aware of the other's work while developing the respective initial manuscripts.

Both works use capped graph traversals to control the cost of reverse sampling, but in different ways. Seddighin develops a phased greedy algorithm that repeatedly samples the residual coverage problem and uses truncated traversals and validation to implement the required sampling. Our algorithm instead fixes $O(\eps k)$ cost-weighted random seeds, stops final reverse searches when they hit this set, and performs one coverage optimization after an independent sample-count estimation phase. In that phase, a capped traversal evaluates a bounded search-cost statistic exactly; a prefix-integral bound connects the same statistic to the expected cost of the final searches.

The roles of AI also differ. In his disclosure, \citet{Seddighin26} states that he developed all algorithms and proofs himself and used GPT and Claude only for proofreading and polishing after writing an initial draft. In contrast, this work was developed through iterative interactions with GPT-6 Astra ultra, including assistance in generating and refining algorithmic ideas and proofs, as well as preparing the manuscript. The author has verified the results and takes full responsibility for the paper's content and correctness.

\subsection{Related Work}
\label{sec:related}
\paragraph{Origins and approximation framework.}
\citet{DR01,RD02} introduced network-aware methods for viral marketing. \citet{KKT03,KKT15} established the combinatorial optimization framework for IC, linear threshold, and more general triggering diffusion, including hardness and submodularity-based approximation results. The underlying cardinality-constrained greedy guarantee is due to~\citet{NWF78}.

\paragraph{Simulation-based accelerations and heuristics.}
CELF, introduced by~\citet{CELF07}, uses lazy marginal-gain evaluation to accelerate greedy selection; CELF++ of~\citet{CELF11} further reduces repeated work. These are accelerations of greedy evaluation, rather than substitutes for its objective. A separate line develops fast heuristic influence surrogates. Representative examples are the degree-discount heuristics of~\citet{CWY09}, the local-arborescence methods of~\citet{CWW10}, and IRIE of~\citet{IRIE12}, which combines influence ranking with influence estimation. These methods emphasize practical scalability; they do not by themselves provide the uniform approximation-and-time guarantee proved here.

\paragraph{Reverse sampling and provable scalability.}
\citet{BBCL14,BBCL16} developed reverse influence sampling, and TIM/TIM+~\citep{TIM14} separates sample-size estimation from coverage optimization. Subsequent work includes martingale-based estimation in IMM~\citep{IMM15}, online quality assessment in OPIM and OPIM-C~\citep{OPIM18}, and improved RR-set generation and tighter complexity bounds by~\citet{GWWC20,GWWLT22}. The latter works improve the graph-dependent sampling cost, including bounds that retain a factor of $k$ under restrictions on incoming probability sums. \citet{Lakshmanan25} also studies budget-independent influence maximization, but the guarantee in that paper's Theorem~1 assumes $k\eps<1$. Our guarantee applies to every $1\le k\le n$ without this restriction.

\paragraph{Comparison with previous approaches.}
The most direct precursor to our stopping rule is HIST, introduced by~\citet[Section~4]{GWWC20} and developed further in~\citet[Section~5]{GWWLT22}. HIST first selects a \emph{sentinel set}, meaning a subset of the eventual seeds used to stop subsequent reverse searches. It then generates RR samples that stop upon hitting this set and selects the remaining seeds. Its sentinel-selection phase uses empirical greedy optimization and sample-based validation to establish an influence guarantee for the selected prefix.

Our method replaces this influence-optimized sentinel construction with $O(\eps k)$ independent cost-weighted draws. The resulting set need not capture any prescribed fraction of the optimal influence: its size controls the approximation loss, while its distribution controls the cost of reverse searches. A prefix-integral argument gives a cost bound inversely proportional to the number of draws plus one. Combining this bound with independent sample-count estimation yields~\eqref{eq:main-time}.

TIM~\citep[Section~3.2]{TIM14} uses indegree-weighted random seed sets to define an auxiliary influence quantity and controls the inverse moment of the resulting lower-bound estimate. Our cost-weighted draws instead serve as actual seeds that reduce subsequent search work. For sample-count estimation, we apply the accumulated-sum method for bounded observations~\citep{DKLR00} to a capped search-cost statistic, rather than explicitly estimating an influence lower bound. This statistic can be evaluated exactly without always completing the RR search. Its mean enters the sample-count bound inversely and the per-sample cost bound linearly; canceling these factors links parameter estimation directly to the computational saving from the prescribed seeds.

\subsection{Technique Overview}
\label{sec:overview}
We explain the construction in the order used in the technical sections: first the seed-selection algorithm for a supplied number of samples, and then the phase that determines that number.

\paragraph{From reverse samples to prescribed seeds.}
A reverse-reachable (RR) set consists of the vertices that can reach a uniformly selected target in one random live-edge graph. A seed set intersects such a sample with probability equal to its influence divided by the number of vertices. Maximizing the fraction of intersected samples is therefore an empirical maximum-coverage problem. Using inverted incidence lists and integer buckets, all greedy selections take time linear in the explicit coverage instance. The difficulty is generating that instance cheaply enough.

As in sentinel-based sampling~\citep{GWWC20,GWWLT22}, we require the output to contain a set $B$ before generating the optimization samples, and stop a reverse search as soon as it hits $B$. Unlike selecting $B$ by empirical influence maximization, we draw at most $\eps k/3$ vertices independently, with probabilities proportional to one plus their indegrees. These weights reflect the cost of expanding a vertex. The membership test precedes the incoming-edge scan, so a costly prescribed vertex stops the search without incurring its own expansion cost. Searches avoiding $B$ are completed, and greedy fills the remaining positions using those samples. Prescribing any set this small loses only $O(\eps)$ in the approximation ratio; no influence-quality certificate for $B$ is needed.

\paragraph{Why the budget factor disappears.}
Our sufficient sample bound still contains a factor of $k$, because we require accuracy simultaneously over all feasible seed sets. We remove this factor from the \emph{total work}, not from that accuracy requirement. Under cost-weighted sampling of $B$, a more expensive search prefix is more likely to contain a prescribed seed. Integrating the probability of avoiding each prefix yields a search-cost bound inversely proportional to the number of prescribed draws plus one. With our rounded $\eps$-fraction of the budget, this contributes $O(1/(\eps k))$. Multiplication with the statistical factor $k$ leaves an additional $1/\eps$, rather than a polynomial budget dependence.

The sample-count estimation phase makes this cancellation valid without knowing the optimum. It measures a bounded statistic of reverse-search cost whose mean controls both the number of final samples and their expected generation cost. A smaller mean allows cheaper searches but calls for more samples; a larger mean has the opposite effect. The unknown mean cancels in their product. Crucially, the preliminary observations are independent of both $B$ and the final sample stream, so conditioning on the chosen sample count preserves the search-cost bound. Accuracy is proved uniformly over all feasible sets before applying it to the data-dependent greedy output.

\paragraph{Computing the statistic without completing the search.}
The statistic stops increasing once the total cost of reached vertices reaches a threshold equal to a $1/k$ fraction of the graph's total search cost. Its value is then known exactly, so the search can stop before expanding the vertex that crosses the threshold. Repeating these computations until their values sum to a confidence threshold determines the final sample count. Each computation is charged in proportion to the value it contributes, giving a worst-case bound on the work of this preliminary phase. This second stopping rule is distinct from hitting $B$: it determines only a scalar needed for sampling, rather than the membership of a complete RR set. Both stopping rules are exact and introduce no additional approximation error.

\subsection{Paper Structure}

Section~\ref{sec:model} defines the input and diffusion models and collects the elementary inequalities. Section~\ref{sec:selection} presents RR sampling, prescribed-seed selection, and approximation for a sufficient sample count. Section~\ref{sec:calibration} determines the count and proves the complete running-time guarantee, including a time-capped implementation. Section~\ref{sec:trigger-extension} extends the analysis beyond IC. Section~\ref{sec:conclusion} concludes.

\section{Preliminaries}
\label{sec:model}
The input to influence maximization consists of a finite directed graph $G=(V,E)$, a specified stochastic diffusion model on $G$, and an integer seed budget $1\le k\le |V|$. We write $n=|V|\ge1$, $m=|E|$, $N^-(v)=\{u:(u,v)\in E\}$, and $\deg^-(v)=|N^-(v)|$. Every vertex is eligible as a seed and every seed has unit cost. A diffusion model specifies how activation spreads from an initially active seed set. For a fixed $S\subseteq V$, let $A(S)$ be the random final active set under that model, including the seeds themselves. The influence-maximization problem is to select a set of at most $k$ seeds maximizing expected final spread:
\begin{equation}
 \sigma(S)=\E[|A(S)|],\qquad
 \OPT_k=\max_{S\subseteq V:\,|S|\le k}\sigma(S).
 \label{eq:objective}
\end{equation}
Since seeds themselves count as active,
\begin{equation}
 k\le\OPT_k\le n.
 \label{eq:opt-range}
\end{equation}

We seek a randomized approximation algorithm that, given accuracy $0<\eps\le1/4$ and failure parameter $0<\delta<1/2$, returns a feasible set whose influence is at least $(1-1/e-\eps)\OPT_k$ with probability at least $1-\delta$. This success probability refers to the algorithm's random choices; the expectation defining $\sigma$ is over diffusion. We first specify the IC model used in the main analysis, with the graph stored in adjacency lists, and then the triggering model used in the extension.

\subsection{Independent cascades and live-edge graphs}
An IC network specifies, in addition to $G$, a probability $p_{uv}\in[0,1]$ for every edge $(u,v)\in E$. At time zero, exactly the vertices in a chosen set $S$ are active. If $u$ becomes active for the first time at round $t$, then in round $t+1$ it has one opportunity to activate each still-inactive out-neighbor $v$, succeeding with probability $p_{uv}$. All such edge trials are mutually independent. Each opportunity is used at most once, and active vertices remain active. The process stops when a round activates no new vertex. Since every nonterminal round activates at least one previously inactive vertex, it stabilizes after at most $n$ rounds. Its final active set is denoted by $A(S)$.

A \emph{live-edge graph} for this IC network is the random directed subgraph $L=(V,E_L)$ obtained by drawing independent variables
\[
 Z_{uv}\sim\operatorname{Bernoulli}(p_{uv}),\qquad (u,v)\in E,
 \qquad E_L=\{(u,v)\in E:Z_{uv}=1\}.
\]
An edge in $E_L$ is called live. For any fixed directed graph $L$, define
\[
 \Reach_L(S)=\{v\in V:\text{there is a directed path in $L$ from some $s\in S$ to $v$}\}.
\]
Paths of length zero are allowed, so $S\subseteq\Reach_L(S)$.

Under the standard live-edge coupling for IC~\citep[Section~4.2]{KKT15}, $A(S)=\Reach_L(S)$, and hence $\sigma(S)=\E_L[|\Reach_L(S)|]$.

\subsection{The triggering model}
\begin{definition}[Triggering diffusion{~\citep[Section~4.1]{KKT15}}]
\label{def:trigger}
For each vertex $v$, a triggering model specifies a probability distribution $D_v$ on $2^{N^-(v)}$. Independently across vertices, draw a set $T_v\sim D_v$ once at the beginning of a diffusion. With $A_0=S$, the active sets evolve as
\[
 A_{t+1}=A_t\cup\{v\in V\setminus A_t:T_v\cap A_t\ne\varnothing\}.
\]
The corresponding live-edge graph has vertex set $V$ and edge set
\[
 E_L=\{(u,v)\in E:u\in T_v\}.
\]
\end{definition}
For fixed triggering sets, induction on $t$ shows that $A_t$ contains exactly the vertices reachable from $S$ by live paths of at most $t$ edges. Thus the final active set is $\Reach_L(S)$, and the objective~\eqref{eq:objective} applies. Dependence among incoming edges of one vertex is allowed; triggering sets at different vertices are independent. IC is the special case in which each $u\in N^-(v)$ is included independently with probability $p_{uv}$.

\subsection{Elementary inequalities}
\label{sec:elementary}
We use a standard additive form of the Chernoff bound. It follows by combining the two one-sided inequalities stated in~\citet[Lemma~1]{TIM14}; we quote it rather than reprove the concentration inequality.

\begin{lemma}[Chernoff bound{~\citep[Lemma~1]{TIM14}}]
\label{lem:chernoff}
Let $Y_1,\ldots,Y_N$ be independent and identically distributed random variables in $[0,1]$ with mean $p$. For every $x>0$,
\begin{equation}
 \Prob\!\left[\left|\frac1N\sum_{i=1}^NY_i-p\right|\ge x\right]
 \le 2\exp\!\left(-\frac{Nx^2}{2p+x}\right).
 \label{eq:chernoff}
\end{equation}
\end{lemma}
The statement includes Bernoulli observations. When $p=0$, all observations are zero almost surely and the inequality is immediate.

The next elementary comparison connects a probability of at least one success with a linear quantity truncated at one. It will be used when selecting the sample count.
\begin{lemma}[Truncated linear comparison]
\label{lem:scalar}
For every $x\in[0,1]$ and integers $1\le s\le k$,
\begin{equation}
 1-(1-x)^s\le\min\{1,kx\},\qquad
 1-(1-x)^k\ge\tfrac12\min\{1,kx\}.
 \label{eq:scalar}
\end{equation}
\end{lemma}
\begin{proof}
The first inequality follows from $1-(1-x)^s\le sx\le kx$ and the upper bound one. For the second, $(1-x)^k\le e^{-kx}$. Concavity of $1-e^{-y}$ on $[0,1]$ and monotonicity for $y\ge1$ give
\[
 1-e^{-y}\ge(1-e^{-1})\min\{1,y\}
            \ge\tfrac12\min\{1,y\}\qquad(y\ge0).
\]
Set $y=kx$ to complete the proof.
\end{proof}

\section{Reverse-Reachable Sampling with Prescribed Seeds}
\label{sec:selection}
This section develops the seed-selection phase. We first recall why RR sampling yields an empirical coverage problem, then explain how a small prescribed seed set permits exact early termination. We analyze approximation assuming a sufficiently large sample count $N$. The only unresolved parameter at the end of the section is how to choose $N$ without knowing $\OPT_k$; Section~\ref{sec:calibration} supplies it.

\subsection{The RR-set approach}
Sample a live-edge realization $L$ and an independent uniform target $t\in V$. Its reverse-reachable set is
\begin{equation}
 R=\RR(L,t)=\{u\in V:u\leadsto t\text{ in }L\}.
 \label{eq:rr}
\end{equation}
A reverse breadth-first search generates $R$ by expanding incoming live edges. The target itself belongs to $R$, so every RR set is nonempty. A seed set intersects $R$ exactly when it can activate $t$ in this realization. This observation underlies reverse influence sampling~\citep{BBCL14,TIM14}. The following identity is due to~\citet{BBCL16}; we include a proof for completeness.

\begin{lemma}[RR hitting identity{~\citep[Observation~3.2]{BBCL16}}]
\label{lem:rr}
For every fixed $S\subseteq V$,
\begin{equation}
 \Prob[R\cap S\ne\varnothing]=\frac{\sigma(S)}n.
 \label{eq:rr-identity}
\end{equation}
\end{lemma}
\begin{proof}
Condition on $L$. The event $R\cap S\ne\varnothing$ is equivalent to $t\in\Reach_L(S)$. Its conditional probability is $|\Reach_L(S)|/n$ because $t$ is uniform. Taking expectation over $L$ proves the claim.
\end{proof}

For $N$ independent RR samples $R_1,\ldots,R_N$, define
\begin{equation}
 \widehat\sigma_N(T)=\frac nN\sum_{i=1}^{N}
              \ind\{R_i\cap T\ne\varnothing\}.
 \label{eq:empirical}
\end{equation}
The ordinary RR algorithm generates the samples, views them as hyperedges, and greedily selects $k$ vertices to cover as many hyperedges as possible. Every selection covers all still-uncovered samples containing the chosen vertex. With sufficiently many samples, approximate maximization of this empirical objective implies approximate maximization of true influence.

\begin{lemma}[Coverage structure]
\label{lem:submodular}
Both $\sigma$ and every fixed empirical objective $\widehat\sigma_N$ are normalized, nonnegative, monotone, and submodular.
\end{lemma}
\begin{proof}
For a fixed realization, $|\Reach_L(S)|$ is the size of the union of the singleton reachability sets of vertices in $S$. If $A\subseteq B$, the vertices newly reached by adding $v$ to $B$ form a subset of those newly reached by adding $v$ to $A$. This proves diminishing marginal gains; the other properties follow directly. Expectation preserves all these properties. Each summand in~\eqref{eq:empirical} is also a coverage indicator, and their nonnegative sum has the same properties.
\end{proof}

The sample count and the cost per sample are separate issues. Even when the explicit coverage instance can be optimized quickly, generating it may require many expensive reverse searches. Our modification changes how much of each sample must be revealed.

\subsection{Presampling seeds and stopping covered searches}
\label{sec:presample}
Instead of starting greedy with no seeds, we first reserve a small number of seed positions for randomly chosen vertices. Let $B$ be the set of distinct vertices obtained. Every eventual output must contain $B$. Consequently, a reverse search that reaches $B$ has already established that its sample is covered by every possible output. It can return a hit indicator and stop the \emph{entire search}. A search that avoids $B$ must instead return its complete RR set.

The random choice should favor vertices whose expansion is expensive. Under IC, expanding $v$ examines its incoming edges, so its cost is proportional to $1+\deg^-(v)$. We therefore use the following distribution.

\begin{definition}[IC search costs and prescribed-seed sampling]
\label{def:ic-costs}
For the input IC network, set
\begin{equation}
 c(v)=1+\deg^-(v),\qquad
 C=\sum_{v\in V}c(v)=m+n,\qquad
 \pi(v)=\frac{1+\deg^-(v)}{m+n}.
 \label{eq:costs}
\end{equation}
Choose
\begin{equation}
 b=\left\lfloor\frac{\eps k}{3}\right\rfloor
 \label{eq:b}
\end{equation}
and draw $b$ vertices independently with replacement from $\pi$. Their set of distinct values is the prescribed seed set $B$. Thus $|B|\le b$, leaving $k-|B|$ positions for greedy selection.
\end{definition}

Algorithm~\ref{alg:absorb} implements the stopped reverse search, denoted $\Absorb$. Its return value $\HIT$ records that the complete sample intersects $B$; we call this an \emph{absorbed sample}. The membership test in \hyperref[line:absorb-membership]{line~\ref*{line:absorb-membership}} precedes expansion, so a costly prescribed vertex can stop the search without having its own incoming edges scanned.

\begin{algorithm}[H]
\caption{$\Absorb(B)$: reveal only the relevant part of an RR sample}
\label{alg:absorb}
\KwIn{The IC network $(G,p)$ and a fixed prescribed set $B$}
\KwOut{$\HIT$, or a complete RR set disjoint from $B$}
Sample a uniform root $t$; initialize a FIFO queue $Q=[t]$ and discovered set $U=\{t\}$\;
\While{$Q\ne\varnothing$}{
 $v\gets\operatorname{pop}(Q)$\;
 \If{$v\in B$\nllabel{line:absorb-membership}}{\Return{$\HIT$}\Comment*[r]{Do not expand $v$}}
 \ForEach{$(u,v)\in E$, in incoming-adjacency order}{
  Draw an independent $Z_{uv}\sim\operatorname{Bernoulli}(p_{uv})$\;
  \If{$Z_{uv}=1$ and $u\notin U$}{Add $u$ to $U$ and append it to $Q$\;}
 }
}
\Return{$U$}\;
\end{algorithm}

For a supplied sample count $N$, $\Select$ (Algorithm~\ref{alg:select}) calls $\Absorb$ (Algorithm~\ref{alg:absorb}) exactly $N$ times, keeps the nonhit sets as $\cR$, and runs coverage greedy on those sets. A hit consumes one trial and is not replaced. The total trial count remains $N$, even though only a subset of the trials need stored vertex lists.

\begin{algorithm}[H]
\caption{$\Select(k,\eps,N)$: seed selection for a supplied sample count}
\label{alg:select}
\KwIn{The IC network and a sampler for $\pi$ from Definition~\ref{def:ic-costs}; $k$, $\eps$, and $N\ge1$}
\KwOut{A seed set $S$ of size $k$}
$b\gets\lfloor\eps k/3\rfloor$\;
Draw $b$ vertices independently with replacement from $\pi$; let $B$ be their set of distinct values\;
$\cR\gets$ an empty multiset\;
\For{$i=1,\ldots,N$}{
 $Y\gets\Absorb(B)$ using a fresh root and fresh independent edge trials\;
 \If{$Y\ne\HIT$}{Append $Y$ to $\cR$\;}
}
$A\gets\Greedy(\cR,B,k-|B|)$\;
\Return{$B\cup A$}\;
\end{algorithm}

\begin{lemma}[Exact compression]
\label{lem:compression}
Couple Algorithm~\ref{alg:absorb} with a complete RR sample $R$. It returns $\HIT$ if and only if $R\cap B\ne\varnothing$, and otherwise returns exactly $R$. Consequently, if $\cR$ is the retained multiset from $N$ trials, then for every $A\subseteq V\setminus B$,
\begin{equation}
 \widehat\sigma_N(B\cup A)=\frac nN
 \left(N-|\cR|+\sum_{R\in\cR}\ind\{R\cap A\ne\varnothing\}\right).
 \label{eq:compression}
\end{equation}
\end{lemma}
\begin{proof}
Fix the live/dead status of every edge, a root, and the incoming-adjacency order of an uninterrupted reverse breadth-first search. Until the first vertex in $B$ is popped, the interrupted and uninterrupted searches have the same queue and discovered vertices. If the complete RR set intersects $B$, the interrupted search reaches its first such vertex and returns $\HIT$. Otherwise it finishes unchanged and returns $R$.

For a hit, every set containing $B$ intersects the complete sample, so its contribution to~\eqref{eq:empirical} is one. For a nonhit, the returned complete set is disjoint from $B$, and its contribution is determined by intersection with $A$. Summing over the $N$ trials proves~\eqref{eq:compression}.
\end{proof}

The number of hits is $N-|\cR|$, a common additive term in every admissible solution's empirical value. It does not affect greedy choices. Thus discarding a hit \emph{record} is safe, but deleting that trial from the denominator or drawing until $N$ nonhits are obtained would change the sampling distribution. The set $B$ remains fixed throughout the batch; it is not enlarged by intermediate greedy choices.

\paragraph{Implementing all greedy selections in linear total time.}
\label{sec:bucket}
It is important not to introduce another factor of $k$ when optimizing the sampled instance. Index the retained sets as $R_1,\ldots,R_M$, with repeated sets assigned different indices, and let
\[
 I=\sum_{j=1}^M|R_j|,\qquad
 \mathcal I(v)=\{j:v\in R_j\}.
\]
For each unselected vertex, maintain its current number $g(v)$ of uncovered incident sets. Integer buckets indexed by $g(v)$ support constant-time decrements and extraction of a maximum-count vertex. Algorithm~\ref{alg:greedy} gives the full update rule.

\begin{algorithm}[H]
\caption{$\Greedy(\cR,B,q)$: linear total-time maximum-coverage greedy}
\label{alg:greedy}
\KwIn{Nonempty retained sets $\cR=(R_1,\ldots,R_M)$, all disjoint from $B$; $q\le n-|B|$}
\KwOut{An additional set $A\subseteq V\setminus B$ of size $q$}
$A\gets\varnothing$; mark vertices of $B$ selected and every retained set uncovered\;
\If{$q=0$}{\Return{$A$}\;}
Build the inverted lists $\mathcal I(v)$ and set $g(v)\gets|\mathcal I(v)|$\;
Initialize buckets $0,\ldots,M$ as doubly linked lists; insert each $v\notin B$ into bucket $g(v)$\;
$D\gets\max_{v\notin B}g(v)$\;
\For{$i=1,\ldots,q$}{
 \While{bucket $D$ is empty and $D>0$}{$D\gets D-1$\;}
 Remove a vertex $x$ from bucket $D$; mark $x$ selected; add $x$ to $A$\;
 \ForEach{$j\in\mathcal I(x)$}{
  \If{$R_j$ is uncovered}{
   Mark $R_j$ covered\;
   \ForEach{$v\in R_j$}{
    \If{$v$ is unselected}{Move $v$ to bucket $g(v)-1$; set $g(v)\gets g(v)-1$\;}
   }
  }
 }
}
\Return{$A$}\;
\end{algorithm}

\begin{lemma}[Coverage implementation]
\label{lem:bucket}
Algorithm~\ref{alg:greedy} makes exact maximum-marginal-gain selections and uses $O(n+I)$ total time and space.
\end{lemma}
\begin{proof}
Initially $g(v)$ counts its uncovered incident sets. When a previously uncovered set is first covered, precisely its unselected vertices have their counts decremented. Thus this invariant is maintained, and a vertex of maximum count has maximum marginal gain. Each bucket move costs $O(1)$ using a membership handle.

Each set is marked covered once, so the total length of all newly covered set scans is $I$. Each chosen vertex has its inverted list scanned once; the sum of those lengths is at most $I$, including entries for already covered sets. Counts never increase, so the maximum bucket pointer moves downward at most $M$ times. Since the retained sets are nonempty, $M\le I$ unless both are zero. Building the lists and buckets, maintaining flags, and making $q\le n$ selections therefore cost $O(n+I)$. When all remaining counts are zero, the zero bucket supplies arbitrary unselected vertices without a scan of the ground set. The same structures use $O(n+I)$ space.
\end{proof}

\subsection{Approximation for a sufficient sample count}
\label{sec:fixed-quality}
We first quantify the loss from prescribing $B$. This argument is deterministic and holds for every small initial set, not only for a typical random one. Randomness in $B$ is needed later for the running-time analysis.

\begin{lemma}[Greedy with prescribed seeds]
\label{lem:forced}
Let $f:2^V\to\mathbb R_{\ge0}$ be normalized, monotone, and submodular. Let $O$ be optimal under budget $k$, and let $b$ be an integer with $0\le b\le k$, and fix $B\subseteq V$ with $|B|\le b$. Starting from $B$, repeatedly take an exact maximum-marginal-gain greedy step until the set has size $k$, obtaining $S$. Then
\begin{equation}
 f(S)\ge\left(1-\frac1e\right)\left(1-\frac bk\right)f(O).
 \label{eq:forced}
\end{equation}
\end{lemma}
\begin{proof}
Set $q=k-b$. The algorithm makes at least $q$ greedy steps because $|B|\le b$. Let $S_0=B$ and let $S_i$ be its set after the first $i$ such steps, for $0\le i\le q$. At a current set $S_i$, monotonicity and submodularity imply
\[
 f(O)-f(S_i)
 \le f(S_i\cup O)-f(S_i)
 \le\sum_{v\in O\setminus S_i}
       \bigl(f(S_i\cup\{v\})-f(S_i)\bigr).
\]
There are at most $k$ summands. The maximum marginal gain is therefore at least $(f(O)-f(S_i))/k$. With $D_i=f(O)-f(S_i)$, greedy gives $D_{i+1}\le(1-1/k)D_i$. For $q\ge1$, iterating yields
\begin{equation}
 f(S_q)\ge
 \left[1-\left(1-\frac1k\right)^q\right]f(O)
 +\left(1-\frac1k\right)^qf(B).
 \label{eq:greedy-recurrence}
\end{equation}
For $k>1$, $(1-1/k)^q\le e^{-q/k}$; the same upper bound holds for $k=q=1$. Since $1-e^{-x}\ge(1-1/e)x$ on $[0,1]$ and $f(B)\ge0$, we obtain $f(S_q)\ge(1-1/e)(q/k)f(O)$. Any additional greedy steps can only increase the value, so $f(S)\ge f(S_q)$ proves~\eqref{eq:forced}. If $q=0$, its right-hand side is zero and the claim follows from nonnegativity.
\end{proof}

By Lemma~\ref{lem:compression}, greedy on the retained samples is precisely greedy for the complete empirical objective starting from $B$. Lemma~\ref{lem:forced} gives
\begin{equation}
 \widehat\sigma_N(S)\ge
 \left(1-\frac1e\right)\left(1-\frac bk\right)
 \max_{|T|\le k}\widehat\sigma_N(T).
 \label{eq:emp-factor}
\end{equation}
The maximizing set on the right need not contain $B$. Although its empirical value need not be accessible from the compressed data, it is well defined by the complete sample coupling.

To transfer~\eqref{eq:emp-factor} to the true objective, we need uniform accuracy over all feasible seed sets. We use the Chernoff inequality from Lemma~\ref{lem:chernoff}.

\begin{theorem}[Sufficient samples for prescribed-seed selection]
\label{thm:fixed-N}
Let $0<\eps\le1/4$, $0<\delta<1/2$, and
\begin{equation}
 H=\left\lceil k\ln\frac{en}{k}+\ln\frac4\delta\right\rceil.
 \label{eq:H}
\end{equation}
For a fixed sample count satisfying
\begin{equation}
 N\ge\frac{27nH}{\eps^2\OPT_k},
 \label{eq:N-sufficient}
\end{equation}
Algorithm~\ref{alg:select} returns a $(1-1/e-\eps)$ approximation with probability at least $1-\delta/2$. This holds for every fixed prescribed set of size at most $b$ chosen independently of the final RR samples.
\end{theorem}
\begin{proof}
Write $\OPT=\OPT_k$. For any fixed $T$ with $|T|\le k$, the indicators in~\eqref{eq:empirical} are independent Bernoulli variables of mean $p_T=\sigma(T)/n\le\OPT/n$. Let $\alpha=\eps/3$. Lemma~\ref{lem:chernoff} gives
\begin{align}
 \Prob\bigl[|\widehat\sigma_N(T)-\sigma(T)|>\alpha\OPT\bigr]
 &\le2\exp\!\left(-\frac{N\alpha^2\OPT^2}
                         {2n\OPT+n\alpha\OPT}\right)\notag\\
 &\le2\exp\!\left(-\frac{N\eps^2\OPT}{27n}\right)
 \le2e^{-H}.
 \label{eq:one-set-tail}
\end{align}
The middle inequality uses $\alpha<1$ to bound the denominator by $3n\OPT$.

There are at most $(en/k)^k$ feasible sets. To verify this also for large $k$, set $x=k/n\le1$; then
\[
 x^k\sum_{i=0}^k\binom ni
 \le\sum_{i=0}^k\binom ni x^i
 \le(1+x)^n\le e^k.
\]
A union bound in~\eqref{eq:one-set-tail} therefore shows that, except with probability at most $2(en/k)^k e^{-H}\le\delta/2$,
\begin{equation}
 |\widehat\sigma_N(T)-\sigma(T)|\le\frac\eps3\OPT
 \quad\text{for every }T\subseteq V\text{ with }|T|\le k.
 \label{eq:uniform}
\end{equation}
On this event, take a true optimal set $O$ and apply~\eqref{eq:emp-factor}:
\begin{align}
 \sigma(S)
 &\ge\widehat\sigma_N(S)-\frac\eps3\OPT\notag\\
 &\ge\left(1-\frac1e\right)\left(1-\frac bk\right)
       \widehat\sigma_N(O)-\frac\eps3\OPT\notag\\
 &\ge\left(1-\frac1e\right)\left(1-\frac bk\right)\OPT
       -\frac{2\eps}{3}\OPT\notag\\
 &\ge\left(1-\frac1e-\eps\right)\OPT.
 \label{eq:quality-transfer}
\end{align}
The third line uses that the coefficient multiplying $\widehat\sigma_N(O)$ is at most one; the last uses $b/k\le\eps/3$. The uniform event covers the data-dependent greedy output, so that output need not be independent of the samples. The proof holds conditionally for any independent fixed $B$ of size at most $b$, and hence also for the random $B$ in Algorithm~\ref{alg:select}.
\end{proof}

This theorem separates the approximation argument from sample-count estimation. The remaining task is to meet~\eqref{eq:N-sufficient} without knowing $\OPT_k$, while controlling the work required to generate those samples.

\section{Determining the Sample Count in Nearly Linear Time}
\label{sec:calibration}
The sufficient condition~\eqref{eq:N-sufficient} involves the unknown optimum. A direct solution is to estimate that optimum or repeatedly optimize larger sample batches. We instead choose a statistic adapted to the cost of the prescribed-seed searches. Its expectation need not closely approximate $\OPT_k$: the essential requirement is that it support both a sufficient sample count and a matching upper bound on search work.

We first derive this statistic from the search-cost bound. We then show how to observe it cheaply, how a scalar stopping rule determines $N$, and why the complete algorithm has the claimed approximation and running time. The preliminary searches used to determine the sample count are independent of the prescribed seeds and the RR samples used for seed selection.

\subsection{A statistic matched to the cost of reverse searches}
\label{sec:matching}
For a complete RR set, define
\begin{equation}
 w(R)=\sum_{v\in R}c(v),\qquad
 X(R)=\min\left\{1,\frac{k\,w(R)}C\right\},\qquad
 \mu=\E_R[X(R)].
 \label{eq:score}
\end{equation}
The statistic is the normalized search cost truncated at one; we call it the \emph{score} of $R$. It is small for samples of small total local cost, and equals one when that cost reaches $C/k$. The algorithm will observe independent copies of $X(R)$ without always revealing all of $R$.

To see why this statistic is useful, consider a complete reverse search with a fixed vertex order. A costly prefix has a large probability of containing a cost-weighted prescribed seed. Therefore the probability of continuing a long search decreases with accumulated cost. Integrating this survival probability gives the following bound.

\begin{lemma}[Prefix-integral bound]
\label{lem:prefix}
Fix a complete search order $v_1,\ldots,v_r$, independent of $B$, and write $a_i=\sum_{j=1}^i c(v_j)$, $a_0=0$. Let $B$ consist of the distinct vertices in $b$ independent draws from $\pi$. Charge $c(v)$ for each vertex expanded before absorption and let $Q_B$ be the total charge. Then
\begin{equation}
 \E_B[Q_B\mid v_1,\ldots,v_r]
 \le\frac C{b+1}
       \left[1-\left(1-\frac{a_r}C\right)^{b+1}\right].
 \label{eq:prefix}
\end{equation}
\end{lemma}
\begin{proof}
Vertex $v_i$ is expanded exactly when all $b$ draws avoid $\{v_1,\ldots,v_i\}$. The prefix includes $v_i$ because membership is tested before expansion. Hence
\begin{align*}
 \E_B[Q_B\mid v_1,\ldots,v_r]
 &=\sum_{i=1}^r c(v_i)\left(1-\frac{a_i}C\right)^b\\
 &\le\sum_{i=1}^r\int_{a_{i-1}}^{a_i}
                            \left(1-\frac xC\right)^b\,dx\\
 &=\frac C{b+1}
       \left[1-\left(1-\frac{a_r}C\right)^{b+1}\right].
\end{align*}
The inequality follows from monotonicity of the integrand. For $b=0$, the zeroth power is interpreted as one, consistent with no absorbing seed.
\end{proof}

\begin{lemma}[Score comparison]
\label{lem:proxy}
The score in~\eqref{eq:score} satisfies
\begin{equation}
 \frac kC\le X(R)\le1,\qquad \frac{n\mu}{2}\le\OPT_k.
 \label{eq:proxy}
\end{equation}
\end{lemma}
\begin{proof}
The RR set is nonempty, so $1\le w(R)\le C$. Since $k\le n\le C$, the bounds on $X(R)$ follow. Independently of $R$, draw $k$ vertices from $\pi$ and let $U_k$ be their distinct set. Conditional on $R$, the probability that $U_k$ hits it is $1-(1-w(R)/C)^k$. Lemma~\ref{lem:scalar} and the RR identity give
\begin{align*}
 \frac{n\mu}{2}
 \le n\E_R\!\left[1-\left(1-\frac{w(R)}C\right)^k\right]=\E_{U_k}[\sigma(U_k)]\le\OPT_k.
\end{align*}
Every possible $U_k$ contains at most $k$ eligible vertices, proving the final inequality.
\end{proof}

The inequality $n\mu/2\le\OPT_k$ gives a sufficient scale for sampling, but it need not be tight. Nor do we assume $n\mu\ge k$. The deterministic cap in the eventual sample schedule will use the separate fact $\OPT_k\ge k$.

For a complete RR set, the bracket in Lemma~\ref{lem:prefix} is at most $X(R)$ when $b+1\le k$. Averaging over the RR sample gives the following expected search-cost bound.
\begin{theorem}[Expected cost of an absorbing sample]
\label{thm:sample-cost}
Assume $b+1\le k$. If the prescribed set $B$ and a fresh RR sample are independent, then
\begin{equation}
 \E_{B,R}[\operatorname{Time}(\Absorb(B))]
 =O\!\left(\frac{C\mu}{b+1}\right).
 \label{eq:absorb-cost}
\end{equation}
Each search has worst-case running time $O(C)=O(m+n)$.
\end{theorem}
\begin{proof}
Condition on a complete realization, its root, and the uninterrupted search order. In Lemma~\ref{lem:prefix}, $a_r=w(R)$. Since $b+1\le k$, Lemma~\ref{lem:scalar} gives
\[
 1-\left(1-\frac{w(R)}C\right)^{b+1}\le X(R).
\]
Thus $\E_{B,R}Q_B\le C\mu/(b+1)$.

Each expanded vertex, including its incoming-edge trials, discovery checks, and queue insertions, costs $O(c(v))$. Root initialization and the final membership check cost $O(1)$. Every nonroot vertex discovered before termination was enqueued by an already charged expansion. Clearing the remaining queue, discarding partial lists, and resetting touched marks can be charged to the same operations. Therefore the actual time is $O(1+Q_B)$. Lemma~\ref{lem:proxy} gives $\mu\ge k/C$, so
\[
 \frac{C\mu}{b+1}\ge\frac{k}{b+1}\ge1.
\]
The additive constant is thus absorbed by $C\mu/(b+1)$ in expectation, proving~\eqref{eq:absorb-cost}. A vertex is expanded at most once, so $Q_B\le C$; this also gives the worst-case bound.
\end{proof}

The expectation in~\eqref{eq:absorb-cost} is joint over the random prescribed set and the new sample. We do not assert this bound for every fixed $B$. The final algorithm will keep $B$ independent of the estimation phase, so the random sample count can be analyzed without conditioning on a favorable prescribed set.

\subsection{Observing the score by an exact capped search}
\label{sec:capped-score}
A complete RR sample may be large, but $X(R)=1$ as soon as the costs of the vertices already removed from the queue sum to at least $C/k$. Every such vertex belongs to the complete RR set, so this threshold determines the score exactly. The test includes the current vertex's precomputed cost but precedes the scan of its incoming edges.

Algorithm~\ref{alg:score} implements this rule. It returns $1$ when the threshold is reached and otherwise completes the search to compute $kw(R)/C$.

\begin{algorithm}[H]
\caption{$\Score(k)$: exact score from a capped reverse search}
\label{alg:score}
\KwIn{The IC network, the costs from Definition~\ref{def:ic-costs}, and budget $k$}
\KwOut{Exactly $X(R)$ for an independent complete RR sample}
Sample a uniform root $t$; initialize $Q=[t]$, $U=\{t\}$, and $w\gets0$\;
\While{$Q\ne\varnothing$}{
 $v\gets\operatorname{pop}(Q)$; $w\gets w+c(v)$\;
 \If{$kw\ge C$}{\Return{$1$}\Comment*[r]{Do not expand the crossing vertex}}
 \ForEach{$(u,v)\in E$, in incoming-adjacency order}{
  Draw an independent $Z_{uv}\sim\operatorname{Bernoulli}(p_{uv})$\;
  \If{$Z_{uv}=1$ and $u\notin U$}{Add $u$ to $U$ and append it to $Q$\;}
 }
}
\Return{$kw/C$}\;
\end{algorithm}

\begin{lemma}[Exact capped score and its work]
\label{lem:score-cost}
Algorithm~\ref{alg:score} returns the exact statistic $X(R)$ under a coupling with a complete RR sample. For a realized complete sample $R$, its running time is bounded by
\begin{equation}
 O\!\left(\frac Ck X(R)\right).
 \label{eq:score-cost}
\end{equation}
\end{lemma}
\begin{proof}
Fix a complete realization and compare with an uninterrupted reverse BFS. Before termination, the two searches have identical discoveries and queue order. If the threshold is crossed, every popped vertex is in the complete RR set, so $w(R)$ is at least the accumulated cost and its score is one. Otherwise the search finishes, visits every vertex of $R$ once, and computes $kw(R)/C<1$.

In the noncrossing case, the sum of expanded local costs is $w(R)=(C/k)X(R)$, which pays for the search and its cleanup. In the crossing case, the expanded vertices have total cost strictly less than $C/k$; the crossing vertex itself is checked but not expanded. Initialization and that last check cost $O(1)$, absorbed because $C/k\ge1$. As in Theorem~\ref{thm:sample-cost}, queued vertices were inserted by charged expansions, which also pay for their cleanup. Thus both cases satisfy~\eqref{eq:score-cost}.
\end{proof}

This is an exact statistic of an ordinary RR sample, not an estimate of a different diffusion process. In particular, repeated calls with fresh randomness produce independent observations with the original mean $\mu$.

\subsection{Accumulating scores to determine the sample count}
\label{sec:sample-count}
The estimation phase accumulates independent values $X_1,X_2,\ldots$ returned by $\Score(k)$. It stops when their sum reaches a prescribed threshold $\Gamma$. This inverse-sum rule is a bounded-observation stopping method~\citep{DKLR00}; Lemma~\ref{lem:estimation} states its guarantees.

For $\Gamma\ge1$, let
\begin{equation}
 \tau=\min\left\{t:\sum_{i=1}^{t}X_i\ge\Gamma\right\}.
 \label{eq:tau}
\end{equation}
The estimation phase returns only $\tau$. It does not return a seed set, an influence estimate, or an estimate of $\OPT_k$.

\begin{lemma}[Sample-count estimation guarantees]
\label{lem:estimation}
The stopping time in~\eqref{eq:tau} satisfies
\begin{align}
 \tau&\le\left\lceil\frac{\Gamma C}{k}\right\rceil
       &&\text{deterministically},\label{eq:tau-cap}\\
 \E\tau&\le\frac{\Gamma+1}{\mu},\label{eq:tau-mean}\\
 \Prob\!\left[\frac{\Gamma}{2\tau}>\mu\right]
 &\le\exp\!\left[-\Gamma\left(\ln2-\tfrac12\right)\right]
 \le e^{-\Gamma/8}.\label{eq:tau-tail}
\end{align}
The estimation phase takes $O(C(\Gamma+1)/k)$ time in the worst case.
\end{lemma}
\begin{proof}
Every observation is at least $k/C$, proving~\eqref{eq:tau-cap}. Since every observation is at most one, the total at termination satisfies
\begin{equation}
 \Gamma\le\sum_{i=1}^{\tau}X_i<\Gamma+1.
 \label{eq:overshoot}
\end{equation}
For each $i$, the event $\{\tau\ge i\}$ depends only on $X_1,\ldots,X_{i-1}$ and is independent of $X_i$. The deterministic bound~\eqref{eq:tau-cap} allows finite summation, yielding
\[
 \E\sum_{i=1}^{\tau}X_i
 =\sum_i\E\bigl[X_i\ind\{\tau\ge i\}\bigr]
 =\mu\sum_i\Prob[\tau\ge i]
 =\mu\E\tau.
\]
Together with~\eqref{eq:overshoot}, this proves~\eqref{eq:tau-mean}; no separate optional-stopping premise is needed.

For the one-sided event, set $s=\lceil\Gamma/(2\mu)\rceil-1$. If $s=0$, the event $\tau<\Gamma/(2\mu)$ is impossible. Otherwise this event implies $\sum_{i=1}^sX_i\ge\Gamma$, whereas $s\mu<\Gamma/2$. For $x\in[0,1]$, convexity gives $2^x\le1+x$, so $\E[2^{X_i}]\le1+\mu\le e^\mu$. Markov's inequality and independence imply
\begin{align*}
 \Prob\!\left[\tau<\frac{\Gamma}{2\mu}\right]
 &\le\Prob\!\left[\sum_{i=1}^{s}X_i\ge\Gamma\right]\\
 &\le2^{-\Gamma}\prod_{i=1}^{s}\E[2^{X_i}]\\
 &\le\exp(-\Gamma\ln2+s\mu)\\
 &\le\exp\!\left[-\Gamma\left(\ln2-\tfrac12\right)\right]
 \le e^{-\Gamma/8}.
\end{align*}
This proves~\eqref{eq:tau-tail} by a fixed-length tail bound for an event containing early stopping, rather than by applying a fixed-sample inequality to a stopped average.

Finally, Lemma~\ref{lem:score-cost} bounds the work of any execution by
\[
 \sum_{i=1}^{\tau}O\!\left(\frac CkX_i\right)
 =O\!\left(\frac Ck(\Gamma+1)\right).
\]
This uses exact capped score generation; unnecessarily completing every preliminary RR search would not satisfy this sharper deterministic bound.
\end{proof}

Set $H$ as in~\eqref{eq:H} and define
\begin{equation}
 \Gamma=\left\lceil8\ln\frac2\delta\right\rceil,
 \qquad
 a_\tau=\min\left\{\frac nk,\frac{4\tau}{\Gamma}\right\},
 \qquad
 N=\left\lceil\frac{27H}{\eps^2}a_\tau\right\rceil.
 \label{eq:schedule}
\end{equation}
The next lemma gives the two properties needed from this choice: a sufficient sample count with high probability, and an unconditional moment bound for the running-time analysis.

\begin{lemma}[Sufficiency and unconditional sample-count bounds]
\label{lem:schedule}
With probability at least $1-\delta/2$ over the estimation phase,
\begin{equation}
 a_\tau\ge\frac n{\OPT_k},\qquad
 N\ge\frac{27nH}{\eps^2\OPT_k}.
 \label{eq:schedule-sufficient}
\end{equation}
Unconditionally,
\begin{equation}
 \E[a_\tau]\le\min\left\{\frac nk,\frac8\mu\right\},
 \qquad
 N\le\left\lceil\frac{27nH}{\eps^2k}\right\rceil
 \quad\text{deterministically}.
 \label{eq:schedule-moments}
\end{equation}
\end{lemma}
\begin{proof}
Let $\cA=\{\Gamma/(2\tau)\le\mu\}$. Lemma~\ref{lem:estimation} and the choice of $\Gamma$ give $\Prob[\cA]\ge1-\delta/2$. On $\cA$,
\[
 \frac{4\tau}{\Gamma}\ge\frac2\mu\ge\frac n{\OPT_k},
\]
where the second inequality uses $n\mu/2\le\OPT_k$. Also $n/k\ge n/\OPT_k$. Taking the minimum preserves this lower bound and proves~\eqref{eq:schedule-sufficient}.

For the moments, $a_\tau\le n/k$ holds deterministically. In addition,
\[
 \E[a_\tau]\le\frac4\Gamma\E\tau
 \le\frac{4(\Gamma+1)}{\Gamma\mu}\le\frac8\mu,
\]
using $\Gamma\ge1$. The upper bound on $N$ follows directly from its cap in~\eqref{eq:schedule}. These expectation bounds include all estimation results, not only those in $\cA$.
\end{proof}

\subsection{The complete algorithm: correctness and running time}
\label{sec:complete}
Algorithm~\ref{alg:main} first determines the sample count, then calls $\Select(k,\eps,N)$ (Algorithm~\ref{alg:select}). The calls to $\Score$, the prescribed seed draws inside $\Select$, and the final RR realizations use independent randomness. Consequently, after conditioning on the estimation phase, the final sample count is fixed and both the prescribed seeds and the final sample stream retain their intended distributions.

\begin{algorithm}[H]
\caption{$\Main(G,k,\eps,\delta)$: the complete algorithm}
\label{alg:main}
\KwIn{An IC network $(G,p)$; $1\le k\le n$, $0<\eps\le1/4$, $0<\delta<1/2$}
\KwOut{A set of at most $k$ seeds}
\If{$k=n$}{\Return{$V$}\;}
Build incoming adjacency lists, compute $c(v)$ and $C=m+n$, and initialize a sampler for $\pi$ and reusable vertex arrays\;
$H\gets\lceil k\ln(en/k)+\ln(4/\delta)\rceil$;
$\Gamma\gets\lceil8\ln(2/\delta)\rceil$\;
$Z\gets0$; $\tau\gets0$\;
\While{$Z<\Gamma$}{
 $Z\gets Z+\Score(k)$ with fresh independent randomness\;
 $\tau\gets\tau+1$\;
}
$a_\tau\gets\min\{n/k,4\tau/\Gamma\}$;
$N\gets\lceil(27H/\eps^2)a_\tau\rceil$\;
\Return{$\Select(k,\eps,N)$ using randomness independent of the entire estimation phase}\;
\end{algorithm}

\paragraph{Correctness.}
\begin{theorem}[Approximation of the complete algorithm]
\label{thm:quality}
For the IC model, Algorithm~\ref{alg:main} returns a $(1-1/e-\eps)$ approximation with probability at least $1-\delta$.
\end{theorem}
\begin{proof}
If $k=n$, the returned set is optimal. Otherwise let $\cA$ be the sample-count sufficiency event in Lemma~\ref{lem:schedule}, whose failure probability is at most $\delta/2$. Conditional on any full estimation result in $\cA$, the sample count meets~\eqref{eq:N-sufficient} and the final samples and prescribed seeds are independent of that estimation phase. Theorem~\ref{thm:fixed-N} therefore bounds conditional approximation failure by $\delta/2$. Averaging over the estimation phase gives
\[
 \Prob[\text{approximation failure}]
 \le\Prob[\cA^c]
   +\E\!\left[\ind\{\cA\}
       \Prob[\text{failure}\mid\text{estimation phase}]\right]
 \le\delta.
\]
Feasibility follows from $|B|+(k-|B|)=k$.
\end{proof}

\paragraph{Expected running time and cancellation of the budget.}
\label{sec:runtime}
Two sources of randomness must be kept separate. The cost of a final search is averaged jointly over $B$ and that search's edge trials, whereas the number of such searches is determined by the independent sample-count estimation phase. This separation permits a factorization of expected work; it does not require the costs of searches sharing $B$ to be mutually independent.

\begin{theorem}[Budget-independent time for IC]
\label{thm:ic-time}
Algorithm~\ref{alg:main} has expected running time
\begin{equation}
 O\!\left(\frac{m+n}{\eps^3}
       \left[\ln\frac{en}{k}+\frac1k\ln\frac4\delta\right]\right).
 \label{eq:ic-time}
\end{equation}
\end{theorem}
\begin{proof}
If $k=n$, returning $V$ costs $O(n)$ and satisfies the bound. Otherwise, set $C=m+n$ and recall $b=\lfloor\eps k/3\rfloor$. We have $b+1>\eps k/3$ and $b+1\le k$: for $k=1$ the latter quantity equals one; for $k\ge2$, use $b+1\le k/12+1\le k$. Thus Theorem~\ref{thm:sample-cost} applies.

Constructing the incoming adjacency lists, computing the costs, and initializing the sampling and vertex-mark structures takes $O(m+n)=O(C)$ time. For a concrete constant-time sampler for $\pi$, create an array in which vertex $v$ occurs $c(v)$ times and draw a uniform array position; its construction takes $O(C)$ time. The vertex-mark arrays are allocated and initialized once. A touched-vertex list resets only the entries changed by an individual search, so later searches do not incur another $O(n)$ initialization cost. The sample-count estimation phase takes $O(C\Gamma/k)$ time in the worst case by Lemma~\ref{lem:estimation}. Drawing $b$ seeds and constructing the membership array of $B$ costs $O(n+k)$, absorbed by $C\ge n$.

Condition on the complete result of sample-count estimation. Now $N$ is fixed, and $B$ and the final complete RR samples retain their independent original distributions. Theorem~\ref{thm:sample-cost}, applied to $B$ and each individual sample, and linearity of expectation give conditional expected search work
\[
 O\!\left(\frac{NC\mu}{b+1}\right).
\]
In every execution, the total retained incidence count $I$ is bounded by the work of generating and storing the samples. Lemma~\ref{lem:bucket} adds $O(n+I)$ for all coverage selections. Taking expectation over the estimation phase therefore gives final-phase expected work
\begin{equation}
 O\!\left(n+\frac{C\mu}{b+1}\E[N]\right).
 \label{eq:runtime-factorization}
\end{equation}
This factorization is valid because $N$ is determined independently of $B$ and the final sample stream. Search costs may share $B$ and need not be mutually independent; we only sum their conditional expectations.

The ceiling in~\eqref{eq:schedule} gives $N\le1+(27H/\eps^2)a_\tau$, while Lemma~\ref{lem:schedule} gives $\mu\E[a_\tau]\le8$. Substitution into~\eqref{eq:runtime-factorization} yields
\[
 O\!\left(n+\frac{C\mu}{b+1}
             +\frac{CH}{\eps^2(b+1)}\right)
 =O\!\left(C+\frac{CH}{\eps^2(b+1)}\right),
\]
where we used $\mu\le1$, $b+1\ge1$, and $C\ge n$. Including initialization and sample-count estimation, the unconditional expected time is
\begin{equation}
 O\!\left(C+\frac{C\Gamma}{k}
             +\frac{CH}{\eps^2(b+1)}\right).
 \label{eq:explicit-time}
\end{equation}
All expectation bounds used here include inaccurate estimation results, so no separate rare-event running-time argument is needed.

Write $L_0=\ln(en/k)+k^{-1}\ln(4/\delta)$. Then $L_0\ge1$, $H=O(kL_0)$, and $\Gamma=O(\ln(2/\delta))$. Since $b+1>\eps k/3$,
\[
 \frac{CH}{\eps^2(b+1)}
 =O\!\left(\frac{CH}{\eps^3k}\right)
 =O(C\eps^{-3}L_0).
\]
The terms $C$ and $C\Gamma/k$ are absorbed by the same bound. Substituting $C=m+n$ proves~\eqref{eq:ic-time}.
\end{proof}

The cancellation can now be read directly from the bounds. The number of samples has expectation at most $O(H/(\eps^2\mu))$, while a search costs $O(C\mu/(b+1))$ in joint expectation. Multiplying removes the unknown mean $\mu$. The remaining linear budget term in $H$ is divided by $b+1>\eps k/3$, leaving $\eps^{-3}$ and logarithmic terms rather than a polynomial dependence on $k$.

\subsection{A deterministic running-time limit}
\label{sec:extensions}
The nearly linear guarantee above is an expected-time bound. Independent attempts with individual time limits yield a deterministic running-time bound while retaining a probabilistic approximation guarantee.

\begin{theorem}[Time-capped IC implementation]
\label{thm:capped}
Let $r=\lceil\log_2(2/\delta)\rceil$. There is an implementation with deterministic running-time bound
\begin{equation}
 O\!\left(\frac{r(m+n)}{\eps^3}
       \left[\ln\frac{en}{k}+\frac1k\ln\frac{8r}{\delta}\right]\right)
 \label{eq:capped-time}
\end{equation}
that always outputs a feasible set and satisfies~\eqref{eq:main-quality}.
\end{theorem}
\begin{proof}
Use failure parameter $\gamma=\delta/(2r)$ in each independent attempt. By Theorem~\ref{thm:ic-time}, the expected time of a complete attempt is at most
\[
 T_0=A\frac{m+n}{\eps^3}
      \left[\ln\frac{en}{k}+\frac1k\ln\frac4\gamma\right]
\]
for a sufficiently large absolute constant $A$ fixed by the implementation. Abort an unfinished attempt after $2T_0$ operations. Each attempt has its own independent preliminary samples, prescribed seeds, and final RR samples. Return the first completed output; if all $r$ attempts time out, return any $k$ vertices.

Markov's inequality bounds each timeout probability by $1/2$, so independence bounds the probability that all attempts time out by $2^{-r}\le\delta/2$. Define each attempt's uncapped output on its full random tape. Its approximation-failure probability is at most $\gamma$. An incorrect completed output belongs to the union of the $r$ uncapped failure events, whose probability is at most $r\gamma=\delta/2$. Thus selecting the first completed output needs no assumption that its completion time and quality are independent. The total failure probability is at most $\delta$ and the total time is $O(rT_0)$, yielding~\eqref{eq:capped-time}. Construction and cleanup of attempt-specific data structures are included in the operation budgets.
\end{proof}

\section{Extension to the triggering model}
\label{sec:trigger-extension}
We now return to the model in Definition~\ref{def:trigger}. An IC expansion examines incoming-edge Bernoulli trials; a triggering-model expansion instead samples the whole set $T_v\sim D_v$ and enumerates its members. The representation of $D_v$ may be much more expensive than an adjacency list, so here the corresponding costs must be specified explicitly.

\begin{definition}[Local sampling access for triggering models]
\label{def:access}
The distributions $D_v$ are given with an implementation that can be preprocessed in time $P$. For every vertex $v$, a known integer $c(v)\ge1$ bounds, up to a uniform constant, the \emph{worst-case} time to sample $T_v$ exactly, enumerate it, check discovery marks, and perform the resulting queue operations. Local samples are independent across vertices, and different searches use fresh independent random choices and uniform roots. In this extension, define
\begin{equation}
 C=\sum_{v\in V}c(v),\qquad \pi(v)=\frac{c(v)}C.
 \label{eq:trigger-costs}
\end{equation}
The local triggering-set sampler is charged through $c(v)$.
\end{definition}
This generalizes Definition~\ref{def:ic-costs}; $C$ need not equal $m+n$. Both search procedures retain their stopping tests \emph{before} sampling $T_v$. Costs and the sampler for $\pi$ are prepared once, and all searches reuse vertex arrays whose touched entries can be reset in time proportional to the discoveries in that search. An array containing $c(v)$ copies of each vertex gives an explicit $O(C)$-time construction of a constant-time sampler for $\pi$.

\begin{theorem}[Locally sampleable triggering models]
\label{thm:trigger}
Under Definition~\ref{def:access}, replacing the IC edge-expansion step by exact triggering-set generation gives the approximation guarantee~\eqref{eq:main-quality} in expected time
\begin{equation}
 O\!\left(P+\frac C{\eps^3}
       \left[\ln\frac{en}{k}+\frac1k\ln\frac4\delta\right]\right).
 \label{eq:general-time}
\end{equation}
It also admits a time-capped implementation with deterministic bound
\begin{equation}
 O\!\left(P+\frac{rC}{\eps^3}
       \left[\ln\frac{en}{k}+\frac1k\ln\frac{8r}{\delta}\right]\right),
 \qquad r=\left\lceil\log_2\frac2\delta\right\rceil,
 \label{eq:trigger-capped}
\end{equation}
and success probability at least $1-\delta$. If $P=O(m+n)$ and each triggering set is generated and enumerated in $O(1+\deg^-(v))$ worst-case time, the IC bounds hold for this model as well.
\end{theorem}
\begin{proof}
Fix all triggering sets. The final active set is reachability in the resulting live-edge graph, proving the same RR identity and coverage structure. Independence across vertices makes sampling $T_v$ when $v$ is expanded equivalent to revealing a set sampled in advance. A search's discoveries and queue order therefore agree with a complete search until either stopping rule fires. The exact-compression and exact-score proofs apply unchanged. In particular, correlation among the incoming edges of one vertex introduces no error.

The prefix-integral proof fixes a complete search order before averaging over the prescribed set $B$. It uses only that sampling a vertex with probability $c(v)/C$ hits a prefix of total cost $a$ with probability $a/C$. Thus the identical proof gives expected search work $O(C\mu/(b+1))$; the constant overhead is absorbed because $\mu\ge k/C$ and $b+1\le k$. Definition~\ref{def:access} pays for every local expansion, and the score search still costs $O((C/k)X(R))$. The comparison $n\mu/2\le\OPT_k$ uses a feasible random set of at most $k$ vertices and so is also unchanged. Here $C\ge n$ and seeds still give $\OPT_k\ge k$.

Fresh searches supply independent observations for sample-count estimation and independent complete final RR samples conditional on that phase. The Chernoff bound, prescribed-seed greedy lemma, sample-count sufficiency, and total-probability argument establish the same approximation guarantee. For time, the proof of Theorem~\ref{thm:ic-time} now gives
\[
 O\!\left(P+C+\frac{C\Gamma}{k}
       +\frac{CH}{\eps^2(b+1)}\right).
\]
Using $b+1>\eps k/3$ and $C\ge n$ as before yields~\eqref{eq:general-time}. For~\eqref{eq:trigger-capped}, perform the static preprocessing once, then apply the independent-attempt construction of Theorem~\ref{thm:capped} to the remaining work with $C$ in place of $m+n$. All attempt-specific work is covered by its time budget, so only $P$ lies outside the repetition factor. Finally, when local generation costs $O(1+\deg^-(v))$, take $c(v)=1+\deg^-(v)$ and substitute $C=m+n$.
\end{proof}

For example, the standard linear-threshold live-edge representation selects at most one incoming neighbor per vertex. Scanning cumulative incoming weights samples the selected neighbor or the empty set in edge-linear local time~\citep[Section~4.3]{KKT15}. Triggering sets may also have positively correlated or mutually exclusive members. An arbitrary distribution over the $2^{\deg^-(v)}$ possible triggering sets can require a much larger input representation. The quantities $P$ and $C$ account for that cost rather than treating an unrestricted distribution as a free sampling oracle.

\section{Conclusion}
\label{sec:conclusion}
We developed an influence-maximization algorithm whose nearly linear expected running time has no multiplicative polynomial dependence on the seed budget. The seed-selection phase reserves $O(\eps k)$ positions for cost-weighted random vertices, uses them to terminate already covered reverse searches exactly, and completes the solution by maximum-coverage greedy. Its approximation analysis retains uniform accuracy over all feasible sets. An independent sample-count estimation phase determines a sufficient sample count without guessing the optimum. Exact truncation makes this phase inexpensive, and matching its statistic to the absorbed-search cost eliminates the unknown scalar and the budget factor from the expected-work bound.

The result applies to arbitrary directed IC networks and to triggering models with the declared local sampling access. It concerns a specified cardinality budget with all vertices eligible and equal seed costs. The worst-case dependence on precision remains $\eps^{-3}$; the analysis does not show that this dependence is necessary. Improving it, extending the budget exchange to restricted seed candidates or other constraints, and producing one approximate ordering for all budgets require further ideas.


\bibliographystyle{plainnat}
\bibliography{budget_independent_im}

@inproceedings{DR01,
  author    = {Pedro Domingos and Matthew Richardson},
  title     = {Mining the Network Value of Customers},
  booktitle = {Proceedings of the Seventh {ACM} {SIGKDD} International Conference on Knowledge Discovery and Data Mining},
  pages     = {57--66},
  publisher = {ACM},
  year      = {2001},
  doi       = {10.1145/502512.502525}
}

@inproceedings{RD02,
  author    = {Matthew Richardson and Pedro Domingos},
  title     = {Mining Knowledge-Sharing Sites for Viral Marketing},
  booktitle = {Proceedings of the Eighth {ACM} {SIGKDD} International Conference on Knowledge Discovery and Data Mining},
  pages     = {61--70},
  publisher = {ACM},
  year      = {2002},
  doi       = {10.1145/775047.775057}
}

@inproceedings{KKT03,
  author    = {David Kempe and Jon Kleinberg and {\'{E}}va Tardos},
  title     = {Maximizing the Spread of Influence through a Social Network},
  booktitle = {Proceedings of the Ninth {ACM} {SIGKDD} International Conference on Knowledge Discovery and Data Mining},
  pages     = {137--146},
  publisher = {ACM},
  year      = {2003},
  doi       = {10.1145/956750.956769}
}

@article{KKT15,
  author    = {David Kempe and Jon Kleinberg and {\'{E}}va Tardos},
  title     = {Maximizing the Spread of Influence through a Social Network},
  journal   = {Theory of Computing},
  volume    = {11},
  number    = {4},
  pages     = {105--147},
  year      = {2015},
  doi       = {10.4086/toc.2015.v011a004}
}

@article{NWF78,
  author    = {George L. Nemhauser and Laurence A. Wolsey and Marshall L. Fisher},
  title     = {An Analysis of Approximations for Maximizing Submodular Set Functions---{I}},
  journal   = {Mathematical Programming},
  volume    = {14},
  pages     = {265--294},
  year      = {1978},
  doi       = {10.1007/BF01588971}
}

@inproceedings{CELF07,
  author    = {Jure Leskovec and Andreas Krause and Carlos Guestrin and Christos Faloutsos and Jeanne VanBriesen and Natalie Glance},
  title     = {Cost-Effective Outbreak Detection in Networks},
  booktitle = {Proceedings of the 13th {ACM} {SIGKDD} International Conference on Knowledge Discovery and Data Mining},
  pages     = {420--429},
  publisher = {ACM},
  year      = {2007},
  doi       = {10.1145/1281192.1281239}
}

@inproceedings{CELF11,
  author    = {Amit Goyal and Wei Lu and Laks V. S. Lakshmanan},
  title     = {{CELF++}: Optimizing the Greedy Algorithm for Influence Maximization in Social Networks},
  booktitle = {Proceedings of the 20th International Conference Companion on World Wide Web},
  pages     = {47--48},
  publisher = {ACM},
  year      = {2011},
  doi       = {10.1145/1963192.1963217}
}

@inproceedings{CWY09,
  author    = {Wei Chen and Yajun Wang and Siyu Yang},
  title     = {Efficient Influence Maximization in Social Networks},
  booktitle = {Proceedings of the 15th {ACM} {SIGKDD} International Conference on Knowledge Discovery and Data Mining},
  pages     = {199--208},
  publisher = {ACM},
  year      = {2009},
  doi       = {10.1145/1557019.1557047}
}

@inproceedings{CWW10,
  author    = {Wei Chen and Chi Wang and Yajun Wang},
  title     = {Scalable Influence Maximization for Prevalent Viral Marketing in Large-Scale Social Networks},
  booktitle = {Proceedings of the 16th {ACM} {SIGKDD} International Conference on Knowledge Discovery and Data Mining},
  pages     = {1029--1038},
  publisher = {ACM},
  year      = {2010},
  doi       = {10.1145/1835804.1835934}
}

@inproceedings{IRIE12,
  author    = {Kyomin Jung and Wooram Heo and Wei Chen},
  title     = {{IRIE}: Scalable and Robust Influence Maximization in Social Networks},
  booktitle = {Proceedings of the 12th {IEEE} International Conference on Data Mining},
  pages     = {918--923},
  publisher = {IEEE Computer Society},
  year      = {2012},
  doi       = {10.1109/ICDM.2012.79}
}

@inproceedings{BBCL14,
  author    = {Christian Borgs and Michael Brautbar and Jennifer Chayes and Brendan Lucier},
  title     = {Maximizing Social Influence in Nearly Optimal Time},
  booktitle = {Proceedings of the 25th Annual {ACM--SIAM} Symposium on Discrete Algorithms},
  pages     = {946--957},
  publisher = {SIAM},
  year      = {2014},
  doi       = {10.1137/1.9781611973402.70}
}

@misc{BBCL16,
  author    = {Christian Borgs and Michael Brautbar and Jennifer Chayes and Brendan Lucier},
  title     = {Maximizing Social Influence in Nearly Optimal Time},
  year      = {2016},
  howpublished = {arXiv:1212.0884v5},
  note      = {Revised full version; cited for the corrected budget-dependent running-time bound},
  url       = {https://arxiv.org/abs/1212.0884v5}
}

@inproceedings{TIM14,
  author    = {Youze Tang and Xiaokui Xiao and Yanchen Shi},
  title     = {Influence Maximization: Near-Optimal Time Complexity Meets Practical Efficiency},
  booktitle = {Proceedings of the 2014 {ACM} {SIGMOD} International Conference on Management of Data},
  pages     = {75--86},
  publisher = {ACM},
  year      = {2014},
  doi       = {10.1145/2588555.2593670}
}

@inproceedings{IMM15,
  author    = {Youze Tang and Yanchen Shi and Xiaokui Xiao},
  title     = {Influence Maximization in Near-Linear Time: A Martingale Approach},
  booktitle = {Proceedings of the 2015 {ACM} {SIGMOD} International Conference on Management of Data},
  pages     = {1539--1554},
  publisher = {ACM},
  year      = {2015},
  doi       = {10.1145/2723372.2723734}
}

@inproceedings{OPIM18,
  author    = {Jing Tang and Xueyan Tang and Xiaokui Xiao and Junsong Yuan},
  title     = {Online Processing Algorithms for Influence Maximization},
  booktitle = {Proceedings of the 2018 International Conference on Management of Data},
  pages     = {991--1005},
  publisher = {ACM},
  year      = {2018},
  doi       = {10.1145/3183713.3183749}
}

@inproceedings{GWWC20,
  author    = {Qintian Guo and Sibo Wang and Zhewei Wei and Ming Chen},
  title     = {Influence Maximization Revisited: Efficient Reverse Reachable Set Generation with Bound Tightened},
  booktitle = {Proceedings of the 2020 {ACM} {SIGMOD} International Conference on Management of Data},
  pages     = {2167--2181},
  publisher = {ACM},
  year      = {2020},
  doi       = {10.1145/3318464.3389740}
}

@article{DKLR00,
  author    = {Paul Dagum and Richard Karp and Michael Luby and Sheldon Ross},
  title     = {An Optimal Algorithm for {Monte Carlo} Estimation},
  journal   = {SIAM Journal on Computing},
  volume    = {29},
  number    = {5},
  pages     = {1484--1496},
  year      = {2000},
  doi       = {10.1137/S0097539797315306}
}

@article{Lakshmanan25,
  author    = {K. Lakshmanan},
  title     = {Influence Maximization Independent of Seed Set Size},
  journal   = {Operations Research Letters},
  volume    = {62},
  pages     = {107309},
  year      = {2025},
  doi       = {10.1016/j.orl.2025.107309}
}

@inproceedings{Wilder18,
  author    = {Bryan Wilder and Laura Onasch-Vera and Juliana Hudson and Jose Luna and Nicole Wilson and Robin Petering and Darlene Woo and Milind Tambe and Eric Rice},
  title     = {End-to-End Influence Maximization in the Field},
  booktitle = {Proceedings of the 17th International Conference on Autonomous Agents and MultiAgent Systems},
  pages     = {1414--1422},
  publisher = {International Foundation for Autonomous Agents and Multiagent Systems},
  year      = {2018},
  url       = {https://www.ifaamas.org/Proceedings/aamas2018/pdfs/p1414.pdf}
}

@article{GWWLT22,
  author    = {Qintian Guo and Sibo Wang and Zhewei Wei and Wenqing Lin and Jing Tang},
  title     = {Influence Maximization Revisited: Efficient Sampling with Bound Tightened},
  journal   = {ACM Transactions on Database Systems},
  volume    = {47},
  number    = {3},
  pages     = {12:1--12:45},
  publisher = {ACM},
  year      = {2022},
  doi       = {10.1145/3533817}
}

@misc{Seddighin26,
  author       = {Saeed Seddighin},
  title        = {Maximizing Social Influence in Almost Linear Time},
  year         = {2026},
  howpublished = {arXiv preprint arXiv:2609.36236},
  eprint       = {2609.36236},
  archivePrefix = {arXiv},
  primaryClass = {cs.DS},
  doi          = {10.48550/arXiv.2609.36236},
  url          = {https://arxiv.org/abs/2609.36236}
}
\end{document}